\documentclass[12pt,english]{article}
\usepackage{mathpazo}
\usepackage[T1]{fontenc}
\usepackage[latin9]{inputenc}
\usepackage{geometry}
\usepackage{amsmath}
\usepackage{amsthm}
\usepackage{amssymb}
\usepackage{stmaryrd}
\usepackage{graphicx}

\makeatletter

\providecommand{\tabularnewline}{\\}

\numberwithin{figure}{section}
\theoremstyle{definition}
\newtheorem{defn}{\protect\definitionname}
\theoremstyle{plain}
\newtheorem{thm}{\protect\theoremname}
\theoremstyle{plain}
\newtheorem{lem}{\protect\lemmaname}
\theoremstyle{remark}
\newtheorem*{acknowledgement*}{\protect\acknowledgementname}

\usepackage{mathrsfs}
\usepackage{amsmath}
\usepackage{mathtools}
\usepackage{braket}
\usepackage{upgreek}
\DeclareMathOperator{\Tr}{Tr}
\DeclareMathOperator{\I}{I}
\DeclareMathOperator{\zero}{O}
\DeclareMathOperator{\Span}{span}
\usepackage{enumitem}
\setenumerate{label=(\roman*)}

\renewcommand{\sigma}{\upsigma}

\makeatother

\usepackage{babel}
\providecommand{\acknowledgementname}{Acknowledgement}
\providecommand{\definitionname}{Definition}
\providecommand{\lemmaname}{Lemma}
\providecommand{\theoremname}{Theorem}

\begin{document}

\title{A Gleason-type theorem for qubits based on mixtures of projective
measurements}

\author{Victoria J Wright and Stefan Weigert\\
 Department of Mathematics, University of York\\
 York YO10 5DD, United Kingdom\\
 \texttt{\small{}vw550@york.ac.uk, stefan.weigert@york.ac.uk}}

\date{August 2018}
\maketitle
\begin{abstract}
We derive Born's rule and the density-operator formalism for quantum
systems with Hilbert spaces of dimension two or larger. Our extension
of Gleason's theorem only relies upon the consistent assignment of
 probabilities to the outcomes of projective measurements and their
classical mixtures. This assumption is significantly weaker than those
required for existing Gleason-type theorems valid in dimension two. 
\end{abstract}
\global\long\def\kb#1#2{|#1\rangle\langle#2|}

\global\long\def\bk#1#2{\langle#1|#2\rangle}

\global\long\def\braket#1#2{\langle#1|#2\rangle}

\global\long\def\ket#1{|#1\rangle}

\global\long\def\bra#1{\langle#1|}

\global\long\def\cd{\mathbb{C}^{d}}

\section{Introduction}

Formulations of quantum theory typically introduce at least three
postulates to define \emph{quantum states} and \emph{observables}
on the one hand, and to explain how they give rise to \emph{measurable
quantities} such as expectation values on the other. One way to set
up the necessary machinery (cf. \cite{Neumann1932}, for example)
consists of postulating that (i) the states\emph{ }of a quantum system
correspond to density operators on a separable, complex Hilbert space
${\cal H}$; (ii) measurements\emph{ }of quantum observables are associated
with collections of mutually orthogonal projection operators acting
on the space ${\cal H}$; (iii) the probabilities of measurement outcomes
are given by Born's rule.\footnote{Normally, these axioms are supplemented by a \emph{measurement postulate}
identifying the post-measurement state once a specific outcome has
been obtained, by a \emph{dynamical law}, and by a rule how to describe
composite quantum systems. Substantially different axiomatic formulations
of quantum theory have been proposed in e.g. \cite{Hardy2001a,Masanes2011}.} In 1957, Gleason \cite{Gleason1957} showed that, assuming the second
postulate, the other two can be seen as a consequence of a quantum
state's most fundamental purpose, that is to assign probabilities
to all measurement outcomes in a consistent way. 

There is, however, a fly in the ointment: Gleason's result only holds
for Hilbert spaces with dimension greater than two. In a two-dimensional
space, the requirement of consistency places no restriction on the
probabilities that may be assigned to non-orthogonal projections.
The resulting surfeit of consistent probability assignments is then
too large to be identified with the set of density operators on $\mathbb{C}^{2}$.
Hence the question: what modification of the assumptions would be
sufficient to recover the probabilistic structure characteristic of
quantum theory in a two-dimensional Hilbert space?

Enter \emph{Gleason-type} theorems, which are designed to fill this
gap. In 2003, the Born rule for Hilbert spaces of dimension two (or
greater) was shown \cite{Busch2003,Caves2004} to follow from extending
Gleason's idea from projection-valued measures (PVMs) to the more
general class of positive operator-valued measures (POMs).

The set of POMs encompasses all quantum\emph{ }measurements, with
PVMs being only a small subset thereof. It is, therefore, natural
to ask whether there are sets ``between'' PVMs and POMs from which
it is possible to derive a Gleason-type theorem. A first step into
this direction was made in 2006 when \emph{three-outcome }POMs were
shown to be sufficient for this purpose in the spaces $\mathbb{C}^{d}$,
with $d\geq2$ \cite{mastersthesis}. Our contribution will take this
reduction even further. We will show that it is possible to derive
a Gleason-type theorem upon extending Gleason's probability assignments
from PVMs to their convex combinations. The resulting \emph{projective-simulable}
measurements \cite{Oszmaniec2017} represent a particularly simple
subset of POMs. 

In Sec. \ref{sec:Gleason-Type-theorems}, we set up our notation and
express Gleason's theorem in a form which is suitable for direct comparison
with Gleason-type theorems. Sec. \ref{sec:main result} describes
projective-simulable POMs in order to derive a Gleason-type theorem
based on assumptions weaker than those currently known. In the final
section, we summarize and discuss our results.

\section{Known extensions of Gleason's theorem \label{sec:Gleason-Type-theorems}}

In this section we review Gleason's theorem and express it in a form
which will allow for easy comparison with later variants, including
our main result. Let us introduce a number of relevant concepts and
establish our notation. 

\subsection{Preliminaries\label{subsec:Preliminaries}}

Let $\mathcal{H}$ be a finite-dimensional, complex Hilbert space.
An \emph{effect} is an operator on $\mathcal{H}$ occurring in the
range of a POM. More explicitly, an effect $e$ is Hermitian and satisfies
$\zero\leq e\leq\I$, where $\zero$ and $\I$ are the zero and identity
operators on $\mathcal{H}$, respectively. Here, the ordering of two
operators, $A\leq B$, say, is defined to hold if the inequality $\bra{\psi}A\ket{\psi}\leq\bra{\psi}B\ket{\psi}$
is satisfied for all elements of the Hilbert space, $\ket{\psi}\in\mathcal{H}$.
We denote the set of all effects $e$ on the space $\mathbb{C}^{d}$
by $\mathcal{E}(\mathbb{C}^{d})\equiv{\cal E}_{d}$.

It is instructive to visualize the effect space of a qubit. The Pauli
operators $\sigma_{x},\sigma_{y}$ and $\sigma_{z}$ with the identity
$\I$ form a basis of Hermitian operators acting on $\mathbb{C}^{2}$.
Hence, any qubit effect takes the form
\begin{equation}
e=a\I+b\sigma_{x}+c\sigma_{y}+d\sigma_{z}\in{\cal E}_{2}\,,\label{eq: effects in C^2}
\end{equation}
where the range of the four parameters $a,\ldots,d\in\mathbb{R}$,
is restricted by the requirement that the operator $e$ must satisfy
$\zero\leq e\leq\I$. 

Fig. \ref{fig: effect space and measurements P and T} illustrates
three-dimensional cross sections of the four-dimensional effect space
obtained upon suppressing the $y$-component in (\ref{eq: effects in C^2}).
The points on the circle in the $xz$-plane correspond to rank-1 projection
operators and are, in addition to $\zero$ and $\I$, the only \emph{extremal}
effects in the sense that they cannot be obtained as convex combinations
of other effects.

\begin{figure}
\begin{centering}
\includegraphics[scale=0.8]{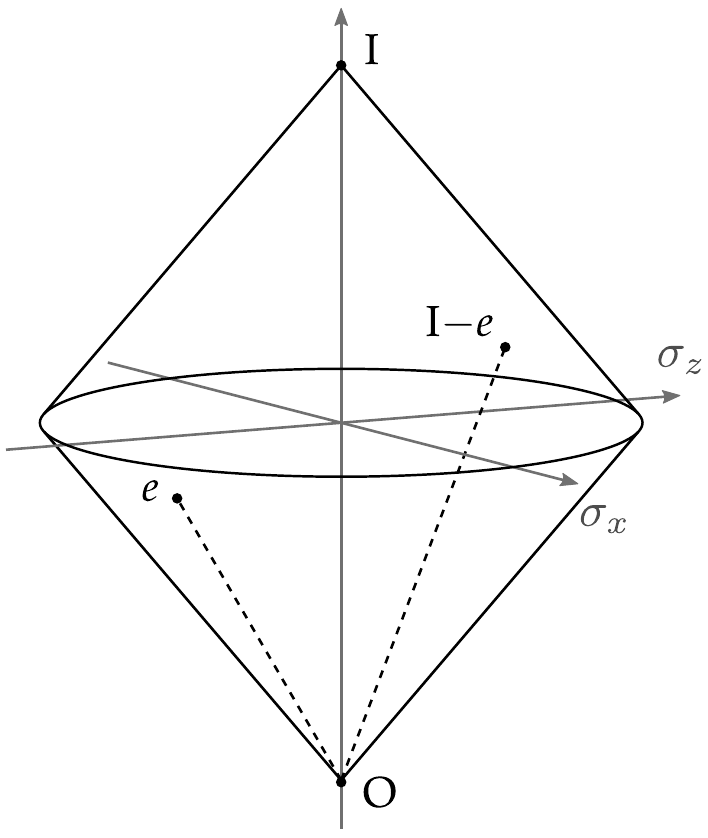}
\par\end{centering}
\caption{\label{fig: effect space and measurements P and T}Three-dimensional
cross section of the four-dimensional qubit effect space, illustrating
a generic two-outcome measurement $\mathbb{D}_{e}$ in (\ref{eq: D_e}),
characterized by the effect $e$.}
\end{figure}

We now turn to the description of measurement outcomes. A fundamental
assumption is that the outcome of any measurement performed on a quantum
system can be associated with some effect $e\in\mathcal{E}_{d}$ \cite{kraus1983states}.
More specifically, we think of a \emph{measurement} $\mathbb{M}$
as an ordered sequence of effects
\begin{equation}
\mathbb{M}=\left\llbracket e_{1},e_{2},\ldots,e_{n}\right\rrbracket ,\quad n\in\mathbb{N}\,,\label{eq: measurement =00003D effect sequence}
\end{equation}
where
\begin{equation}
\sum_{j=1}^{n}e_{j}=\I\,.\label{eq: completeness of effects}
\end{equation}
We say the effect $e_{j}$ is associated with the $j$-th outcome
of the measurement. It is useful to note that measurements with $n$
outcomes $\mathbb{M}=\left\llbracket e_{1},e_{2},\ldots,e_{n}\right\rrbracket $
can be thought of as \emph{vectors }since they are elements of the
vector space formed by the Cartesian product of $n$ copies of the
real vector space of Hermitian operators on $\mathcal{H}$. Hence,
real linear combinations of measurements are well-defined from a mathematical
point of view, though only convex combinations necessarily correspond
to measurements. 

Next, we introduce the concept of a \emph{measurement set} $\mathbf{M}=\left\{ \mathbb{M}_{j},j\in J\right\} $,
for some (possibly uncountable) indexing set $J$, simply consisting
of a collection of selected measurements $\mathbb{M}_{j}$. The set
$\mathbf{M}$ is said to define a particular \emph{measurement scenario}
if we consider (only) the measurements contained in $\mathbf{M}$
to be realisable. The measurement set $\mathbf{PVM}_{d}$ in a Hilbert
space of dimension $d\geq2$, for example, collects all projective
measurements; it thus consists of all measurements of the form $\left\llbracket P_{1},P_{2},\ldots\right\rrbracket $,
with at most $d$ distinct projection operators $P_{j}$ on $\mathbb{C}^{d}$,
i.e. effects satisfying $P_{j}^{2}=P_{j}$. The set $\mathbf{PVM}_{d}$
defines the von Neumann measurement scenario. 

The \emph{effect space} $\mathcal{E}(\mathbf{M})\subseteq{\cal E}_{d}$
consists of all effects which figure in the measurement set $\mathbf{M}$.
In a given scenario, not every effect defined on the space ${\cal H}$
necessarily represents a measurement outcome, thus the set $\mathcal{E}(\mathbf{M})$
may be a proper subset of all effects on the space ${\cal H}$. For
example, in the von Neumann scenario the effect space ${\cal E}(\mathbf{PVM}_{d})$
consists solely of projection operators. The largest possible measurement
set on a Hilbert space with dimension $d$ is given by $\mathbf{POM}_{d}$,
with the only requirement on a measurement $\mathbb{M}=\left\llbracket e_{1},e_{2},\ldots,e_{n}\right\rrbracket \in\mathbf{POM}_{d}$
being that the effects $e_{j}$ satisfy Eq. (\ref{eq: completeness of effects}),
so that indeed $\mathcal{E}(\mathbf{POM}{}_{d})={\cal E}_{d}$. In
other words, every effect will figure in some measurement in a scenario
which considers all POMs to be realisable. In general, POMs may have
infinitely many outcomes. For our considerations, however, those with
only finitely many outcomes will be sufficient.

After these preliminaries, we are ready to describe the role of a
quantum state in a measurement scenario: it should map each effect
$e\in\mathcal{E}(\mathbf{M})$ in the corresponding effect space to
a probability in such a way that the probabilities of all the outcomes
in each measurement $\mathbb{M}\in\mathbf{M}$ sum to one. Such a
map is known as a frame function \cite{Caves2004}.
\begin{defn}
\label{frame function}Let $\mathcal{E}(\mathbf{M})$ be the effect
space associated with the measurement set $\mathbf{M}$. A \emph{frame
function} $f$ in this measurement scenario is a map $f:\mathcal{E}(\mathbf{M})\rightarrow\left[0,1\right]$
such that 
\begin{equation}
\sum_{e_{j}\in\mathbb{M}}f\left(e_{j}\right)=1,\label{eq:frame function def}
\end{equation}
for all measurements $\mathbb{M}$ in the set $\mathbf{M}$.
\end{defn}
We will say that the frame function $f$ \emph{respects the measurement
set }$\mathbf{M}$ if it consistently assigns probabilities to all
effects present in the measurement scenario defined by the set ${\cal \mathbf{M}}$.
Structurally, frame functions resemble probability measures which
quantify the size of disjoint subsets of a sample space, say, with
a relation similar to (\ref{eq:frame function def}) expressing normalization. 

As discussed by Caves et al. \cite{Caves2004}, this approach is intrinsically
non-contextual. When associating outcomes from distinct measurements
with the same mathematical object, we are prescribing that they must
occur with the same probability for a system in a given state, regardless
of context, i.e. which measurements are being performed (see also
\cite{Bell1966}).

\subsection{Gleason's theorem}

Gleason's theorem conveys a limitation of the form which frame functions
may take in a Hilbert space of dimension larger than two. Using the
concepts just introduced, the theorem can be expressed as follows. 
\begin{thm}[Gleason \cite{Gleason1957}]
\label{thm: Gleason }Any frame function $f$ respecting the measurement
set $\mathbf{M}=\mathbf{PVM}_{d}$, $d\geq3$, admits an expression
\begin{equation}
f\left(e\right)=\Tr\left(\rho e\right),\label{eq:gleason-type def}
\end{equation}
for some density operator $\rho$ on $\mathcal{H}$, and all effects
$e\in\mathcal{E}(\mathbf{PVM}_{d})$.
\end{thm}
Originally, Gleason's theorem was stated in terms of measures $\mu$
acting on closed subspaces of the space $\mathcal{H}$. For a countable
collection of mutually orthogonal closed subspaces $\left\{ \mathcal{H}_{1},\mathcal{H}_{2},\ldots,\mathcal{H}_{N}\right\} $
which span the entire space, a measure must satisfy
\begin{equation}
\mu\left(\Span\left\{ \mathcal{H}_{1},\mathcal{H}_{2},\ldots,\mathcal{H}_{N}\right\} \right)=\sum_{j=1}^{N}\mu\left({\cal H}_{j}\right)\,.\label{eq:measure}
\end{equation}
If the dimension of the space $\mathcal{H}$ is at least three, then
any such measure $\mu$ with $\mu\left(\mathcal{H}\right)=1$, necessarily
derives from a density operator $\rho$ on $\mathcal{H}$, via $\mu\left({\cal H}_{j}\right)=\Tr\left(\rho P_{j}\right)$,
$j=1\ldots N$, where the operator $P_{j}$ is the projection onto
the subspace ${\cal H}_{j}$. As shown in Appendix A, Theorem \ref{thm: Gleason }
is equivalent to the original statement of Gleason's theorem.

If two measurements share an effect we will say---following Gleason---that
they \emph{intertwine}. In Hilbert spaces of dimension greater than
two the value of a frame function on any two projections is related
through measurements which intertwine. This relationship then paves
the way for Gleason's theorem. In contrast, projective measurements
on $\mathbb{C}^{2}$ do not intertwine which means that a frame function
may assign probabilities freely to any two non-orthogonal projections.
This freedom allows for frame functions that do not derive from the
trace rule, such as Eq. (\ref{eq:gleason 2 counterexample}) in Sec.
\ref{subsec: Minimal assumptions} below. If, however, one considers
POMs, measurements also intertwine in dimension two. The consequences
of this fact will be seen in the next section.

\subsection{Gleason-type theorems}

Let us now turn to \emph{Gleason-type theorems}, the main topic of
this paper. They are variants of Theorem \ref{thm: Gleason } based
on measurement sets $\mathbf{M}$ different from $\mathbf{PVM}_{d}$.
The resulting, larger effect spaces $\mathcal{E}(\mathbf{M})$ allow
one to extend Gleason's theorem to the case of a qubit and to derive
the result (\ref{eq:gleason-type def}) in a simpler way.

The first Gleason-type theorem was obtained by Busch \cite{Busch2003},
using the measurement set $\mathbf{M}=\mathbf{POM}_{d}$. This is
the most general measurement scenario containing all possible POMs,
and hence has the largest possible effect space, ${\cal E}(\mathbf{POM}_{d})$.
\begin{thm}[Busch \cite{Busch2003}]
 \label{thm: Busch}Any frame function f respecting the measurement
set $\mathbf{M}=\text{\ensuremath{\mathbf{POM}}}_{d}$, $d\geq2$,
admits an expression 
\begin{equation}
f\left(e\right)=\Tr\left(e\rho\right),\label{eq:busch density-2-2}
\end{equation}
for some density operator $\rho$ on $\mathcal{H}$, and all effects
$e\in\mathcal{E}(\mathbf{POM}_{d})\equiv\mathcal{E}_{d}$.
\end{thm}
The assumptions of this theorem are indeed stronger than those of
Theorem \ref{thm: Gleason } because probabilities are assigned to
all \emph{effects,} not just collections of mutually orthogonal projections
in the space $\mathbb{C}^{d}$. Busch required that \emph{generalised
probability measures} $v:\mathcal{E}_{d}\rightarrow\left[0,1\right]$
would need to satisfy the constraints $v(e_{1}+e_{2}+\ldots)=v(e_{1})+v(e_{2})+\ldots$,
for any sequence of effects which may occur in a POM with any number
of outcomes, i.e. $e_{1}+e_{2}+\ldots\leq\I$. This condition is easily
shown to be equivalent to the assumptions in Theorem \ref{thm: Busch}.

The proof of the Theorem \ref{thm: Busch} differs conceptually from
the one given by Gleason. The additivity of frame functions with respect
to any two effects $e_{1}$ and $e_{2}$ occurring in a single measurement
$\mathbb{M}$,
\begin{equation}
f(e_{1}+e_{2})=f(e_{1})+f(e_{2})\,,\label{eq: additivity}
\end{equation}
forces the frame function to be\emph{ homogeneous} for rational numbers,
$f\left(qe\right)=qf\left(e\right)$, $q\in\mathbb{Q}$. Combining
additivity with positivity, $f(e)\geq0$, a frame function is, furthermore,
seen to be homogeneous for \emph{real} numbers, $f(\alpha e)=\alpha f(e)$,
$\alpha\in\mathbb{R}$, and hence is necessarily linear. Extending
this expression linearly from effects to arbitrary Hermitian operators
is consistent only with frame functions given by the trace expression
(\ref{eq:busch density-2-2}). The proof also works in separable Hilbert
spaces of infinite dimension.

An alternative proof of Busch's Gleason-type theorem was given by
Caves et al. \cite{Caves2004}. Instead of showing that frame functions
must be homogeneous, Caves et al. establish their \emph{continuity,}
first at the effect $\zero$, and then for all effects. This property
implies, of course, that frame functions must be linear functions
of effects. 

Revisiting Gleason's theorem, Granstr{\"o}m \cite{mastersthesis} proceeds
along the lines of Busch and Caves et al. when rephrasing the proof.
Interestingly, she only uses POMs with at most \emph{three} outcomes:
the measurement set is given by $\mathbf{M}=\mathbf{3POM}_{d}$ where
any 
\begin{equation}
\mathbb{M}=\left\llbracket e_{1},e_{2},e_{3}\right\rrbracket \in\mathbf{3POM}_{d}\label{eq: three-outcome}
\end{equation}
is a collection of at most three effects. Granstr{\"o}m's observation
is important since her derivation is based on a considerably smaller
measurement set than the one required for the earlier Gleason-type
theorems. 

The following section shows an even smaller measurement set is sufficient
to derive a Gleason-type theorem in dimensions $d\geq2$. The reduction
is not only quantitative but also represents a conceptual simplification
since only POMs arising from classical mixtures of projective measurements,
known as projective-simulable measurements, will be required. 

\section{Assigning probabilities to mixtures of projections\label{sec:main result}}

\subsection{Projective-simulable measurements\label{subsec:Projective-simulable-POVMs}}

\emph{Projective-simulable measurements} (PSMs) are specific POMs
which can be realized by performing projective measurements and combining
them with classical protocols \cite{Oszmaniec2017}. The relevant
classical procedures are given by probabilistically mixing projective
measurements and post-processing of measurement outcomes. Hence, the
experimental implementation of projective-simulable---or \emph{simulable},
for brevity---measurements is not more challenging than that of projective
measurements. In the following, we will suppress any post-processing
since it can always be eliminated by working with suitable mixtures
of measurements (see Lemma 1 in \cite{Oszmaniec2017}). It is important
to note that not all POMs are simulable \cite{Oszmaniec2017}; thus
they represent a proper, non-trivial subset of all POMs.

We will now introduce some two- and three-outcome measurements of
a qubit which are projective-simulable. These are the only measurements
necessary to derive the Gleason-type theorem of Sec. \ref{subsec: PSM Gleason-type-theorem}
when $d=2$. To begin, any (non-trivial) projective qubit measurement
with two outcomes takes the form $\mathbb{M}=\left\llbracket P_{+},P_{-}\right\rrbracket $,
with projections $P_{+}$ and $P_{-}\equiv\I-P_{+}$ on orthogonal
one-dimensional subspaces of the space $\mathbb{C}^{2}$. For example,
on a spin-$\frac{1}{2}$ particle the measurements implemented by
a Stern-Gerlach apparatus oriented along the $x$- or the $z$- axis
would be represented by 
\begin{alignat}{2}
\mathbb{M}_{x}=\frac{1}{2}\left\llbracket \I+\sigma_{x},\I-\sigma_{x}\right\rrbracket  & \quad\text{and}\quad & \mathbb{M}_{z}=\frac{1}{2}\left\llbracket \I+\sigma_{z},\I-\sigma_{z}\right\rrbracket \,,\label{eq:x and z POVMs}
\end{alignat}
respectively. Now imagine a device which performs $\mathbb{M}_{x}$
with probability $p\in[0,1]$ and $\mathbb{M}_{z}$ with probability
$\left(1-p\right)$. The statistics produced by this apparatus are,
in general, no longer described described by a PVM but by a POM, namely
by
\begin{equation}
\mathbb{M}_{xz}(p)=p\mathbb{M}_{x}+\left(1-p\right)\mathbb{M}_{z}\,.\label{eq:x and z mixture}
\end{equation}
Consequently, the POM

\begin{equation}
\mathbb{M}_{xz}(p)=\frac{1}{2}\left\llbracket \I+p\sigma_{x}+\left(1-p\right)\sigma_{z},\I-p\sigma_{x}-\left(1-p\right)\sigma_{z}\right\rrbracket \label{eq: explicit Mxz(p)}
\end{equation}
is projective-simulable since only a probabilistic mixture of projective
measurements is required to implement it. Mixing the simulable measurement
$\mathbb{M}_{xz}(p)$ with another projective or simulable measurement
would result in yet another simulable measurement. 

Clearly, this procedure can be lifted to a Hilbert space with dimension
$d$: mixing any pair of projective or simulable measurements $\mathbb{M}$
and $\mathbb{M}^{\prime}$ with the same number of outcomes, say,
produces another simulable measurement represented by $\mathbb{L}(p)=p\mathbb{M}+(1-p)\mathbb{M}^{\prime}$.
In low dimension such as $d=2$ and $d=3$, the set of $n$-outcome
POMs which can be reached in this way has been characterized in terms
of semi-definite programs \cite{Oszmaniec2017}. 

In our context, the following result for POMs with two outcomes will
be important. 
\begin{lem}
\label{Lemma: two outcome}Let $\mathcal{H}$ be a Hilbert space with
finite dimension $d$. For any effect $e\in\mathcal{E}_{d}$ the two-outcome
POM 
\begin{equation}
\mathbb{D}_{e}=\left\llbracket e,\I-e\right\rrbracket \label{eq: D_e}
\end{equation}
is projective-simulable, i.e. $\mathbb{D}_{e}\in{\cal E}(\mathbf{PSM}_{d})$.
\end{lem}
\begin{proof}
Let $e\in\mathcal{E}_{d}$ be an effect with eigenvalues $\lambda_{j}\in\left[0,1\right]$,
$j=1\ldots d$, labeled in ascending order, i.e. $\lambda_{j}\leq\lambda_{j+1}$.
Being a Hermitian operator, the spectral theorem implies that the
effect $e$ can be written as a linear combination
\begin{equation}
e=\sum_{j=1}^{d}\lambda_{j}P_{j}\,,\label{eq:spectral theorem}
\end{equation}
where $P_{j}\in\mathcal{E}_{d}$ are rank-1 projections onto mutually
orthogonal subspaces of $\mathcal{H}$. Defining the projectors $Q_{k}=\sum_{j=k}^{d}P_{j}$
and letting $p_{k}=\left(\lambda_{k}-\lambda_{k-1}\right)\geq0$ for
$k=1\ldots d$, where $\lambda_{0}\equiv0$, we may rewrite Eq. (\ref{eq:spectral theorem})
as
\begin{equation}
e=\sum_{k=1}^{d}\left(\lambda_{k}-\lambda_{k-1}\right)Q_{k}=\sum_{k=1}^{d}p_{k}Q_{k}\,.\label{eq: Q-projector decomposition}
\end{equation}
This expression for the effect $e$ can be found in \cite{Caves2004}. 

Next, consider the $(d+1)$ projective measurements $\mathbb{P}_{j}=\left\llbracket Q_{j},\I-Q_{j}\right\rrbracket $,
$j=0\ldots d$, where $Q_{0}=\zero$. The choice $p_{0}=\left(1-\lambda_{d}\right)$
ensures that the $(d+1)$ non-negative numbers $p_{j}$ correspond
to probabilities, and satisfy $\sum_{j=0}^{d}p_{j}=1.$ A mixture
of measurements, in which $\mathbb{P}_{j}$ is performed with probability
$p_{j}$, then simulates the desired POM in (\ref{eq: D_e}) since
we have 
\begin{equation}
\sum_{j=0}^{d}p_{j}\mathbb{P}_{j}=\left\llbracket \sum_{j=0}^{d}p_{j}Q_{j},\sum_{j=0}^{d}p_{j}\left(\I-Q_{j}\right)\right\rrbracket =\left\llbracket \sum_{j=1}^{d}p_{j}Q_{j},\I-\sum_{j=1}^{d}p_{j}Q_{j}\right\rrbracket =\left\llbracket e,\I-e\right\rrbracket \,,\label{eq: general 2-outcome}
\end{equation}
which completes the proof.
\end{proof}
Any simulable two-outcome measurement such as $\mathbb{D}_{e}$ can
be used to define simulable three-outcome measurements via $\left\llbracket \zero,e,\I-e\right\rrbracket $
or $\left\llbracket e,\zero,\I-e\right\rrbracket $, for example,
simply by including the effect $\zero$ associated with an outcome
which will never occur. This observation allows us to easily introduce
further simulable three-outcome measurements as probabilistic mixtures. 
\begin{lem}
\label{Lemma three-outcome} Let $\mathcal{H}$ be a Hilbert space
with finite dimension $d$. For any effects $e$ and $e^{\prime}$
with $e+e^{\prime}\in\mathcal{E}_{d}$, the three-outcome POMs 
\begin{equation}
\mathbb{T}_{e}=\left\llbracket \frac{e}{2},\frac{e}{2},\I-e\right\rrbracket \quad\mbox{and}\quad\mathbb{T}_{e,e^{\prime}}=\left\llbracket \frac{e}{2},\frac{e^{\prime}}{2},\I-\frac{\left(e+e^{\prime}\right)}{2}\right\rrbracket \label{eq: three outcomes}
\end{equation}
 are projective-simulable, i.e. $\mathbb{T}_{e}\text{ and }\mathbb{T}_{e,e^{\prime}}\in{\cal E}(\mathbf{PSM}_{d})$. 
\end{lem}
\begin{proof}
The measurement $\mathbb{T}_{e}$ can be obtained from an equal mixture
of two three-outcome measurements, 
\begin{equation}
\mathbb{T}_{e}=\frac{1}{2}\left\llbracket e,\zero,\I-e\right\rrbracket +\frac{1}{2}\left\llbracket \zero,e,\I-e\right\rrbracket \,,\label{eq: Te as mixture}
\end{equation}
each of which is a padded copy of the simulable two-outcome measurement
$\mathbb{P}_{e}$. A slight modification of this argument shows that
the measurement $\mathbb{T}_{e,e^{\prime}}$ corresponds to an equal
probabilistic mixture of two simple simulable three-outcome measurements,
viz.,
\begin{equation}
\mathbb{T}_{e,e^{\prime}}=\frac{1}{2}\left\llbracket e,\zero,\I-e\right\rrbracket +\frac{1}{2}\left\llbracket \zero,e^{\prime},\I-e^{\prime}\right\rrbracket \,.\label{eq: T' as mixture-1}
\end{equation}
\end{proof}
Finally, we would like to point out that in dimension $d=2$, the
measurement set $\mathbf{3PSM}_{2}$, which consists of all three-outcome
simulable POMs, is an eight-parameter family strictly smaller than
$\mathbf{3POM}_{2}$, the set of all all POMs with three outcomes.
For example, the three-outcome POM
\begin{equation}
\mathbb{E}=\frac{1}{3}\left\llbracket \I+\sigma_{x},\I-\frac{1}{2}\sigma_{x}+\frac{\sqrt{3}}{2}\sigma_{z},\I-\frac{1}{2}\sigma_{x}-\frac{\sqrt{3}}{2}\sigma_{z}\right\rrbracket \label{eq:SDP POVM}
\end{equation}
is \emph{not} projective-simulable, which can be verified via the
semi-definite program provided in \cite{Oszmaniec2017}.

\subsection{A Gleason-type theorem based on PSMs\label{subsec: PSM Gleason-type-theorem}}

We will now state and prove our main result, a Gleason-type theorem
derived from projective-simulable measurements.
\begin{thm}[Projective-simulable measurements]
 \label{thm: PSM Gleason-type }Any frame function $f$ respecting
the measurement set $\mathbf{M}=\text{\ensuremath{\mathbf{PSM}}}_{d}$,
$d\geq2$, admits an expression 
\begin{equation}
f\left(e\right)=\Tr\left(e\rho\right),\label{eq:busch density-1-1}
\end{equation}
for some density operator $\rho$ on $\mathcal{H}$, and all effects
$e\in\mathcal{E}\left(\mathbf{PSM}_{d}\right)\equiv\mathcal{E}_{d}$
.
\end{thm}
To prove this theorem, we will show that consistently assigning probabilities
to projective-simulable measurements entails a probability assignment
consistent with all POMs. In other words, a frame function $f$ respecting
the measurement set $\mathbf{PSM}_{d}$ necessarily respects the measurement
set $\mathbf{POM}_{d}$, at which point we can invoke Theorem \ref{thm: Busch}.
\begin{proof}
In a first step, we show that the probability assignments to the effects
$e$ and $e/2$, for any $e\in{\cal E}_{d}$, are not independent.
According to Lemmas \ref{Lemma: two outcome} and \ref{Lemma three-outcome},
the measurements $\mathbb{D}_{e}=\left\llbracket e,\I-e\right\rrbracket $
and $\mathbb{T}_{e}=\left\llbracket e/2,e/2,\I-e\right\rrbracket $
are projective-simulable. By the definition of a frame function $f$
given in (\ref{eq:frame function def}), the probabilities assigned
to the outcomes of these two measurements sum to one, 
\begin{equation}
f\left(e\right)+f\left(\I-e\right)=1=f\left(\frac{e}{2}\right)+f\left(\frac{e}{2}\right)+f\left(\I-e\right).\label{eq:main 2}
\end{equation}
Hence, for any effect $e\in\mathcal{E}_{d}$, we must have 
\begin{equation}
f\left(\frac{e}{2}\right)=\frac{1}{2}f\left(e\right)\,.\label{eq:main half}
\end{equation}

Next, we show that the frame function must be additive for any two
effects $e,e^{\prime}\in\mathcal{E}_{d}$ such that $e+e^{\prime}\in\mathcal{E}_{d}$.
Using Lemma \ref{Lemma: two outcome} again, with $e=(e+e^{\prime})/2$,
we find that the two-outcome measurement 
\begin{equation}
\mathbb{D}_{\frac{1}{2}(e+e^{\prime})}=\left\llbracket \frac{1}{2}\left(e+e^{\prime}\right),\I-\frac{1}{2}\left(e+e^{\prime}\right)\right\rrbracket \label{eq: M(e1+e2)}
\end{equation}
is simulable with projective measurements. Assigning probabilities
to the outcomes of the measurements $\mathbb{D}_{\frac{1}{2}(e+e^{\prime})}$
and $\mathbb{T}_{e,e^{\prime}}$ defined in Eq. (\ref{eq: three outcomes})
is only consistent if the constraint
\begin{equation}
f\left(\frac{1}{2}\left(e+e^{\prime}\right)\right)=f\left(\frac{1}{2}e\right)+f\left(\frac{1}{2}e^{\prime}\right)\,,\label{eq:additive 1}
\end{equation}
is satisfied. Due to the (limited) homogeneity of the frame function
stated in Eq. (\ref{eq:main half}), we conclude that it must be additive,
\begin{equation}
f\left(e+e^{\prime}\right)=f\left(e\right)+f\left(e^{\prime}\right)\,,\label{eq:additive 2}
\end{equation}
on all effects $e$ and $e^{\prime}$ such that $e+e^{\prime}\in\mathcal{E}_{d}$.

Now consider any $n$-outcome measurement $\mathbb{M}=\left\llbracket e_{1},e_{2},\ldots,e_{n}\right\rrbracket $
on $\mathbb{C}^{d}$, for $n\in\mathbb{N}$. Using (\ref{eq:additive 2})
repeatedly and recalling the normalization (\ref{eq: completeness of effects})
of effects, we find by induction that 
\begin{equation}
\begin{aligned}\sum_{j=1}^{n}f\left(e_{j}\right) & =f(e_{1}+e_{2})+\sum_{j=3}^{n}f\left(e_{j}\right)=\ldots=f\left(\sum_{j=1}^{n}e_{j}\right)=f(\I)=1\,.\end{aligned}
\label{eq:busch frame function}
\end{equation}
Hence, any frame function $f$ respecting $\mathbf{PSM}_{d}$ is seen
to respect the measurement set $\mathbf{POM}_{d}$, consisting of
all POMs. Therefore, by Theorem \ref{thm: Busch}, the frame function
must take the form $f\left(e\right)=\Tr\left(e\rho\right)$, for some
density operator $\rho$ on $\mathcal{H}$, and all effects $e\in\mathcal{E}_{d}$,
which is the content of Theorem \ref{thm: PSM Gleason-type }.
\end{proof}
The theorem just proved provides a weakening of the assumptions made
by Busch and Caves et al. in Theorem \ref{thm: Busch}: since the
set of measurements considered is smaller, fewer restrictions are
put on potential frame functions---but exactly the same functions
are recovered. 

\subsection{Minimal assumptions for a Gleason-type theorem\label{subsec: Minimal assumptions}}

We now address the problem of identifying the smallest measurement
set in dimension two from which a Gleason-type theorem maybe be derived.
Recall that Gleason's theorem does not hold in dimension two; frame
functions which respect $\mathbf{PVM}_{2}$ but do not stem from a
density operator (cf. Eq. (\ref{eq:gleason-type def})) are easy to
construct. For instance, assign probabilities to all rank-1 projectors---corresponding
to the points of the Bloch sphere---according to the rule
\begin{equation}
g\left(P\right)=\begin{cases}
0 & \text{ if }P=\ket 0\bra 0\,,\\
1 & \text{ if }P=\ket 1\bra 1\,,\qquad P\in{\cal E}(\mathbf{PVM}_{2})\\
\frac{1}{2} & \text{ otherwise}\,,
\end{cases}\label{eq:gleason 2 counterexample}
\end{equation}
in addition to $g(\zero)=0$ and $g(\I)=1.$ Then, for each projective
measurement $\mathbb{P}=\left\llbracket P,\I-P\right\rrbracket $
we find that the constraint (\ref{eq:frame function def}) on frame
functions is satisfied,
\begin{equation}
g(P)+g(\I-P)=1\,.\label{eq: d=00003D2 constraint on ff}
\end{equation}
Other probability assignments not admitting the desired trace form
can be found in \cite{Hall2016}, for example. These constructions
succeed since, for $d=2$, each projector $P$ occurs \emph{only in
one condition }of the\emph{ }form (\ref{eq: d=00003D2 constraint on ff})
i.e. there are no intertwined measurements.

Similarly frame functions defined on the measurements set $\mathbf{2POM}_{2}$,
the four-parameter family (\ref{eq: effects in C^2}) of POMs for
$\mathbb{C}^{2}$ with at most two outcomes, do not yield a Gleason-type
theorem. Extending the domain of the function $g$ in (\ref{eq:gleason 2 counterexample})
to all effects in $\mathcal{E}_{2}$ results in a frame function that
respects $\mathbf{2POM}_{2}$ but is not of the desired form. Thus,
measurements with three or more outcomes are a necessity in a set
from which a Gleason-type theorem may be proved. Theorem \ref{thm: PSM Gleason-type }
considers one such case, namely the set of projective-simulable measurements
$\mathbf{PSM}_{2}$ having $\mathbf{2POM}_{2}$ as a proper subset.

Could the measurement set $\mathbf{PSM}_{2}$ be the smallest sufficient
set? Looking back at the proof of Theorem \ref{thm: PSM Gleason-type }
given in the previous subsection, it becomes clear that only elements
of $\mathbf{PSM}_{d}$ with at most three outcomes, or those contained
in the set $\mathbf{3PSM}_{d}$, are necessary for the result to hold.
Furthermore, not all elements of the measurement set $\mathbf{3PSM}_{2}$
have been used. While \emph{all} two-outcome POMs $\mathbb{P}_{e}\in\mathbf{2POM}_{2}$
feature, the only simulable three-outcome POMs required are of the
form $\mathbb{T}_{e}$ or $\mathbb{T}_{e,e^{\prime}}$, defined in
(\ref{eq: three outcomes}). However, not all three-outcome simulable
POMs fall into one of these categories. For example, the three-outcome
measurement 
\begin{equation}
\mathbb{T}^{\prime}=\frac{1}{4}\left\llbracket \I+\sigma_{z},\I+\sigma_{x},2\I-\left(\sigma_{z}+\sigma_{x}\right)\right\rrbracket ,\label{eq:three outcome counter example}
\end{equation}
is simulable but does \emph{not} have the form of either $\mathbb{T}_{e}$
or $\mathbb{T}_{e,e^{\prime}}$. Thus, we have actually shown a result
slightly stronger than Theorem \ref{thm: PSM Gleason-type } since,
for $d=2$, we can replace the measurement set $\mathbf{M}$ on which
frame functions need to be defined by 
\begin{equation}
\mathbf{3PSM}_{2}^{\prime}=\mathbf{2POM}_{2}\cup\left\{ \mathbb{T}_{e},\mathbb{T}_{e,e^{\prime}}|e,e^{\prime}\in\mathcal{E}_{d}\text{ such that }e+e^{\prime}\in\mathcal{E}_{d}\right\} \,,\label{eq: def of Min}
\end{equation}
which is a proper subset of the measurement set $\mathbf{3PSM}_{2}\equiv\mathbf{3POM}_{2}\cap\mathbf{PSM}_{2}$,
i.e. all simulable measurements with three outcomes.

We conclude the discussion of ``minimal'' measurement sets by summarizing
the relationship between the sets sufficient to derive a Gleason-type
theorem for a qubit,
\begin{equation}
\mathbf{3PSM}_{2}^{\prime}\subset\left(\mathbf{3POM}_{2}\cap\mathbf{PSM}_{2}\right)\subset\mathbf{3POM}_{2}\subset\mathbf{POM}_{2}.\label{eq:inclusion}
\end{equation}
Fig. \ref{fig:mmt hierarchy} also depicts the insufficient subsets
of two-outcome projections $\mathbf{2PVM}_{2}$ and two-outcome POMs
denoted by $\mathbf{2POM}_{2}$. 

\begin{figure}
\centering{}\includegraphics[scale=0.18]{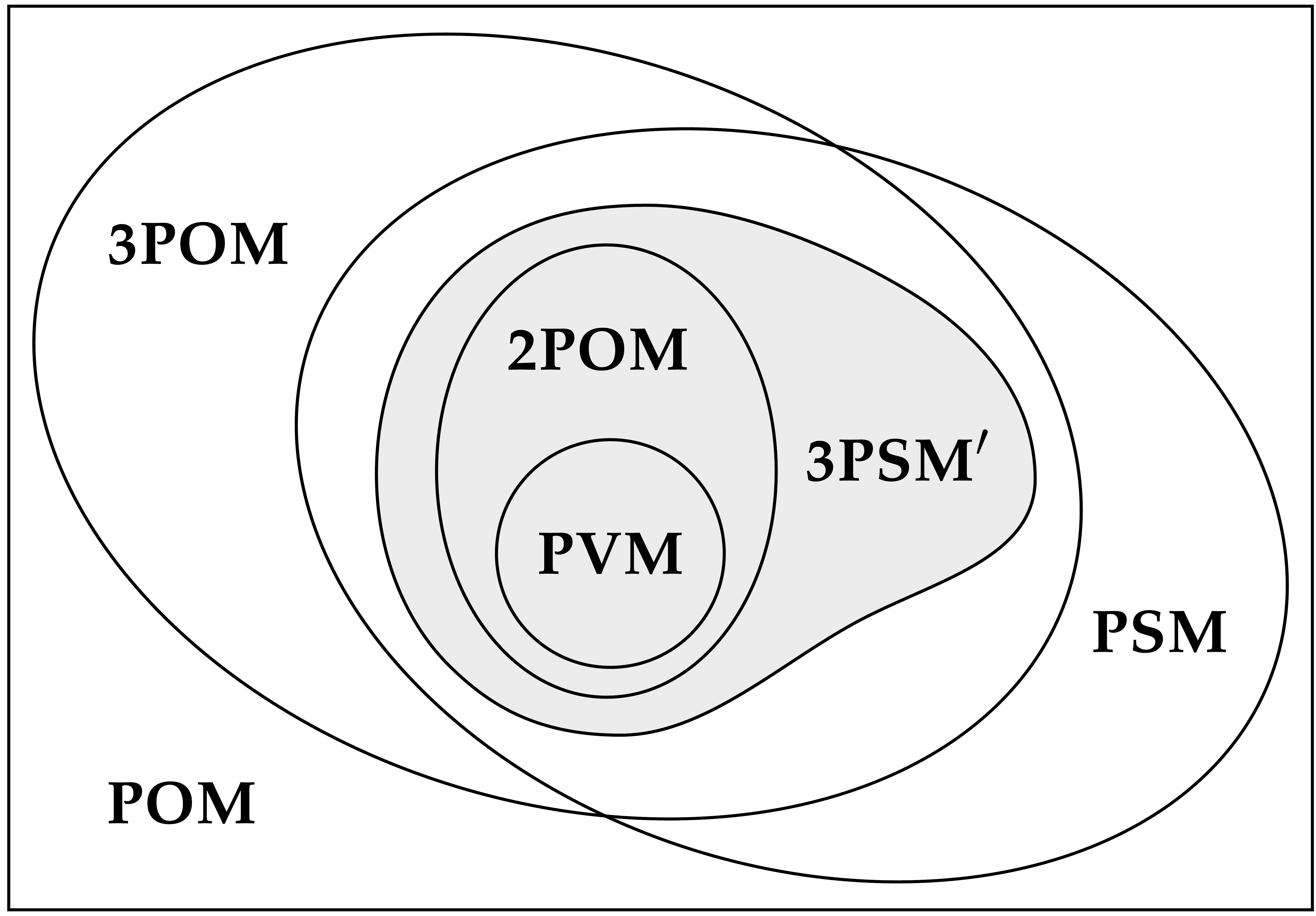}\caption{\label{fig:mmt hierarchy}Supersets and subsets of the set $\mathbf{3PSM}_{2}^{\prime}$
(grey) given in (\ref{eq: def of Min}), the smallest measurement
set known to entail a Gleason-type theorem for a qubit: it strictly
contains the set $\mathbf{2POM}_{2}$ of all two-outcome POMs (cf.
Eq. (\ref{eq:inclusion})) and is strictly contained by the set $\mathbf{3PSM}_{2}$
of all simulable three-outcome POMs; for clarity, the index $2$ has
been dropped from all measurements sets. }
\end{figure}

It is not excluded that measurement sets contained within (or partly
overlapping with) $\mathbf{3PSM}_{2}^{\prime}$ exist which would
still entail a Gleason-type theorem for qubits. In \cite{Caves2004}
frame functions respecting the single measurement (\ref{eq:SDP POVM})
have been shown to admit an expression as in Eq. (\ref{eq:gleason-type def}),
but the result depends the assumption that the frame functions be
continuous on the set of all effects in $\mathcal{E}_{2}$. Hence
this result does not constitute a Gleason-type theorem under our specification.

\subsection{Mixtures and Boolean lattices \label{subsec:Conceptual}}

We now consider how Theorem \ref{thm: PSM Gleason-type } can be interpreted
in view of Hall's discussion \cite{Hall2016} of Busch's Gleason-type
theorem, i.e. Theorem \ref{thm: Busch}. Hall reviews the reasons
which led Gleason (following the work of von Neumann and Birkhoff
\cite{vNBirkhofflogic1936} and Mackey \cite{mackey1957}) to consider
frame functions that respect the measurement set $\mathbf{PVM}_{d}$
consisting of \emph{projective} measurements. Namely, a collection
of mutually orthogonal projections forms a Boolean lattice, thus making
these projections natural candidates to represent disjoint outcomes
of an experiment. General collections of effects which sum to the
identity, on the other hand, do not have this property (see \cite{lahti1995partial},
for example); therefore, a similar justification for considering the
measurement set $\mathbf{POM}_{d}$ cannot be given.

This reasoning also applies to the setting of Theorem \ref{thm: PSM Gleason-type }
since the measurement set $\mathbf{PSM}_{d}$ (or the subset $\mathbf{3PSM_{2}^{\prime}}$)
contains operators other than projections. Nevertheless, the fact
that these measurement sets are made from simulable measurements lends
some support to motivating the additivity of frame functions. 

Gleason's original argument does not work for a qubit because the
constraints (\ref{eq:frame function def}) on frame functions which
result from the measurement set $\mathbf{PVM}_{2}$, are too weak.
If one wishes to derive Born's rule in the space $\mathbb{C}^{2}$,
it is necessary to consider measurement sets larger than $\mathbf{PVM}_{2}$,
thereby invalidating the link between measurements and Boolean lattices.
A particularly simple modification of the measurement set consists
of including convex combinations of the original projective measurements
in $\mathbf{PVM}_{2}$. If one interprets these convex combinations
as classical mixtures of projective measurements then one does not
make statements about other genuinely quantum mechanical measurements
which would lie beyond those of $\mathbf{PVM}_{2}$. 

Let us now make explicit all assumptions which are needed so that
our main result, Theorem \ref{thm: PSM Gleason-type }, may be used
to recover the standard description of states and outcome probabilities
of quantum theory. Importantly, similar---if not stronger---assumptions
must be made in order to achieve the same goal using the Gleason-type
theorems by Busch and Caves at al. 

The first assumption is that there exist projective measurements,
i.e. measurements whose outcomes may be represented by mutually orthogonal
projections on a Hilbert space. Secondly, we assume that it is possible
to perform classical mixtures of measurements, that is to say, given
a pair of measurements $\mathbb{M}$ and $\mathbb{M}^{\prime}$ then
there exists a procedure in which $\mathbb{M}$ is performed with
probability $p$ and $\mathbb{M}^{\prime}$ with probability $\left(1-p\right)$
for any $p\in\left[0,1\right]$. These assumptions alone are \emph{not}
sufficient to restrict states to being represented by density operators. 

To uncover the additional assumption which is needed to implement
our Gleason-type theorem let us consider the procedure just described
in the case of a qubit. For example, we may consider an equal mixture
$\mathbb{M}_{xz}(1/2)$ of the measurements $\mathbb{M}_{x}=\left\llbracket x_{+},x_{-}\right\rrbracket $
and $\mathbb{M}_{z}=\left\llbracket z_{+},z_{-}\right\rrbracket $
from Eq. (\ref{eq:x and z POVMs}) and a mixture $\mathbb{M}_{rs}(p_{+})$
of 
\begin{equation}
\begin{aligned}\mathbb{M}_{r}= & \left\llbracket r_{+},r_{-}\right\rrbracket \,\qquad r_{\pm}=\frac{1}{2}\left(\I\pm\frac{1}{2}\left(\sigma_{x}+\sqrt{3}\sigma_{z}\right)\right)\,,\\
\mathbb{M}_{s}= & \left\llbracket s_{+},s_{-}\right\rrbracket \,\qquad s_{\pm}=\frac{1}{2}\left(\I\pm\frac{1}{2}\left(\sigma_{x}-\sqrt{3}\sigma_{z}\right)\right)\,,
\end{aligned}
\label{eq:alternative POVMs-1-1}
\end{equation}
with probabilities $p_{\pm}=\left(1\pm1/\sqrt{3}\right)/2$, respectively.
Now let us work out the probabilities of the first outcomes of the
measurements $\mathbb{M}_{xz}(1/2)$ and $\mathbb{M}_{rs}(p_{+})$
resulting from the probability assignments given in Eq.\emph{ }(\ref{eq:gleason 2 counterexample}).
\begin{table}
\centering{}%
\begin{tabular}{l|c|c}
 & Probability of outcome 1 & Probability of outcome 2\tabularnewline
\hline 
$\mathbb{M}_{x}$, $\mathbb{M}_{r}$, $\mathbb{M}_{s}$ & $1/2$ & $1/2$\tabularnewline
\hline 
$\mathbb{M}_{z}$ & $0$ & $1$\tabularnewline
\hline 
\end{tabular}\caption{\label{tab:mmt probs}The probabilities of the outcomes of measurements
$\mathbb{M}_{x}$ and $\mathbb{M}_{z}$ in Eq. (\ref{eq:x and z POVMs})
as well as $\mathbb{M}_{r}$ and $\mathbb{M}_{s}$ in Eq. (\ref{eq:alternative POVMs-1-1})
arising from the probability assignment in Eq. (\ref{eq:gleason 2 counterexample}).}
\end{table}
Using the values given in Table \ref{tab:mmt probs}, we find that
for the mixture $\mathbb{M}_{xz}(1/2)$, in which measurements $\mathbb{M}_{x}$
and $\mathbb{M}_{z}$ are performed with equal probability, outcome
one is obtaining with probability
\begin{equation}
\frac{1}{2}g\left(x_{+}\right)+\frac{1}{2}g\left(z_{+}\right)=\frac{1}{4}\,,\label{eq:  equal x and z mixture}
\end{equation}
while the first outcome of $\mathbb{M}_{rs}(p_{+})$ occurs with probability
\begin{equation}
\frac{1}{2}\left(1+\frac{1}{\sqrt{3}}\right)g\left(r_{+}\right)+\frac{1}{2}\left(1-\frac{1}{\sqrt{3}}\right)g\left(s_{+}\right)=\frac{1}{2}\,.
\end{equation}
Not surprisingly, different mixtures of different projective measurements,
which correspond to unassociated processes with well-defined outcome
probabilities, may result in different outcome probabilities. 

According to quantum theory, however, the two mixtures just considered
necessarily give rise to the same outcome probabilities for any qubit
state and thus may both be represented by the same pair of effects,
namely

\begin{equation}
\frac{1}{2}\left(\mathbb{M}_{x}+\mathbb{M}_{z}\right)=p_{+}\mathbb{M}_{r}+p_{-}\mathbb{M}_{s}=\left\llbracket m,\I-m\right\rrbracket \equiv\mathbb{D}_{m}\label{eq: identity}
\end{equation}
with the effect
\begin{equation}
m=\frac{1}{2}\left(\I+\frac{1}{2}\left(\sigma_{x}+\sigma_{z}\right)\right)\,,\label{eq:mixed effect}
\end{equation}
as illustrated in Fig. \ref{fig:mixture of mmts}. 
\begin{figure}
\centering{}\includegraphics[scale=0.9]{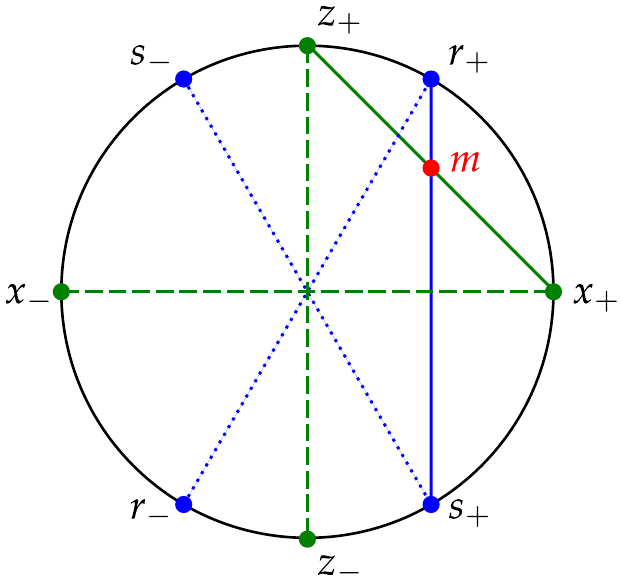}\caption{\label{fig:mixture of mmts}Example of an effect which represents
an outcome stemming from two different mixing procedures: the effect
$m$ occurs as the first outcome of both $\left(\mathbb{M}_{x}+\mathbb{M}_{z}\right)/2$
and $p_{+}\mathbb{M}_{r}+p_{-}\mathbb{M}_{s}$; straight dashed (green)
and dotted (blue) lines connect the pairs of effects in the same measurement,
and the straight solid lines represent the effects which can be formed
by mixing the effects they connect. }
\end{figure}

Thus we see that to exclude $g(P)$ of (\ref{eq:gleason 2 counterexample})
as a valid frame function, it is sufficient to assume that a mixture
of projective measurements $\left\{ \mathbb{M}_{1},\mathbb{M}_{2},\ldots\right\} $,
with probabilities $\left\{ p_{1},p_{2},\ldots\right\} $, is associated
with the convex combination $\left(p_{1}\mathbb{M}_{1}+p_{2}\mathbb{M}_{2}+\ldots\right)$.
The ensuing assignment of effects from $\mathcal{E}\left(\mathbf{POM}_{d}\right)$
to represent outcomes of mixtures is our third assumption and results
in a theory with effect space $\mathcal{E}_{d}$ and measurement set
$\mathbf{PSM}_{d}$. When combining this requirement with Theorem
\ref{thm: PSM Gleason-type }, frame-function arguments become sufficiently
strong to imply Born's rule and the standard density-operator formalism
of quantum theory in the space $\mathbb{C}^{d}$.

\section{Summary and discussion\label{sec:Summary}}

This paper improves on Gleason-type theorems which aim to extend Gleason's
result to Hilbert spaces of dimension $d=2$. The goal is to recover
Born's rule and the representation of quantum states as density operators
as a product of consistent probability assignments to measurement
outcomes. Our main result, given by Theorem \ref{thm: PSM Gleason-type },
shows that any consistent assignment of probabilities to the outcomes
of projective-simulable measurements, or the measurement set $\mathbf{PSM}_{d}$,
must be associated with a density operator in the desired way. Moreover,
we show that a smaller set of measurements $\mathbf{3PSM}_{d}^{\prime}$,
defined in Eq. (\ref{eq: def of Min}), also has this property.

Our result improves upon existing Gleason-type theorems which are
based either on probability assignments to POMs with any number of
outcomes (which constitute the set $\mathbf{POM}_{d}$, see \cite{Busch2003,Caves2004})
or those with at most three outcomes (which constitute the set $\mathbf{3POM}_{d}$,
see \cite{mastersthesis}). The measurement set we consider, $\mathbf{3PSM}_{d}^{\prime}$,
is a strict subset of $\mathbf{3POM}_{d}$. Fig. \ref{fig:mmt hierarchy}
summarizes the relationship between the sets of measurements.

In addition to these quantitative improvements, Theorem \ref{thm: PSM Gleason-type }
also provides new qualitative insights. Projective-simulable measurements
are conceptually simpler than arbitrary POMs because they are just
classical mixtures of projective measurements, with an equal level
of experimental feasibility. Due to the limitation to simulable measurements,
our Gleason-type theorem resembles Gleason's original theorem more
strongly than its predecessors. Furthermore, in Sec. \ref{subsec:Conceptual}
we add an explicit assumption to the setting of Gleason's original
theorem in order to extend the result to dimension two. This assumption
consists of identifying those measurements which, whilst arising from
different mixtures, are known to be indistinguishable in ordinary
quantum theory.

Future work will show whether the subset $\mathbf{3PSM}_{2}^{\prime}$
of projective-simulable measurements, on which the proof of Theorem
\ref{thm: PSM Gleason-type } relies, is the smallest possible set
from which a Gleason-type theorem may be derived in dimension $d=2$.
We cannot exclude that the frame functions respecting $\mathbf{3PSM}_{2}^{\prime}$
are still \emph{overdetermined} in the sense that other sets not containing
all of $\mathbf{3PSM}_{2}^{\prime}$ may also entail a Gleason-type
theorem for a qubit. \\

\begin{acknowledgement*}
\end{acknowledgement*}
Paul Busch (1955-2018) agreed to look at a draft of this paper but
it was not meant to be. We dedicate this paper to the memory of our
kind colleague and wise friend. 

The authors would like to thank Leon Loveridge for helpful discussions
and comments on the manuscript. VW gratefully acknowledges funding
from the York Centre for Quantum Technologies.

\appendix

\section{An equivalent form of Gleason's theorem}

Theorem \ref{thm: Gleason } is equivalent to Gleason's original theorem.
To see this, consider a collection $\left\{ \mathcal{H}_{1},\mathcal{H}_{2},\ldots\right\} $
of mutually orthogonal, closed subspaces of $\mathcal{H}$. Then there
exists a closed subspace $\mathcal{H}^{\perp}$ orthogonal to each
$\mathcal{H}_{j}$ such that $\Span\left\{ \mathcal{H}_{1},\mathcal{H}_{2},\ldots;\mathcal{H}^{\perp}\right\} =\mathcal{H}$.
Using projectors $P_{j}$ and $P^{\perp}$ onto these subspaces, we
have, by the definition of a frame function, that
\begin{equation}
\begin{aligned}\sum_{j}f\left(P_{j}\right)+f\left(P^{\perp}\right)=1 & =f\left(P_{\Span\left\{ \mathcal{H}_{1},\mathcal{H}_{2},\ldots\right\} }\right)+f\left(P^{\perp}\right)\\
 & =f\left(\sum_{j}P_{j}\right)+f\left(P^{\perp}\right),
\end{aligned}
\label{eq:additive frame function 1}
\end{equation}
which gives 
\begin{equation}
\sum_{j}f\left(P_{j}\right)=f\left(\sum_{j}P_{j}\right).\label{eq:additive frame function 2}
\end{equation}
This relation implies that any frame function $f$ respecting $\mathbf{PVM}_{d}$
defines a measure $\mu$ on the closed subspaces $\mathcal{C}$ of
$\mathcal{H}$ given by $\mu\left(\mathcal{C}\right)=f\left(P_{\mathcal{C}}\right)$
since 
\begin{equation}
\begin{aligned}\mu\left(\Span\left\{ \mathcal{H}_{1},\mathcal{H}_{2},\ldots\right\} \right) & =f\left(P_{\Span\left\{ \mathcal{H}_{1},\mathcal{H}_{2},\ldots\right\} }\right)=\sum_{j}\mu\left(\mathcal{H}_{j}\right)\,.\end{aligned}
\label{eq:frame function defines measure}
\end{equation}
Conversely, if $\mu$ satisfies Equation (\ref{eq:measure}) and $\mu\left(\mathcal{H}\right)=1$,
then we have 
\begin{equation}
\sum_{j}\mu\left(\mathcal{H}_{j}\right)+\mu\left(\mathcal{H}^{\perp}\right)=\mu\left(\Span\left\{ \mathcal{H}_{1},\mathcal{H}_{2},\ldots;\mathcal{H}^{\perp}\right\} \right)=\mu\left(\mathcal{H}\right)=1\,,\label{eq:measure defines frame function}
\end{equation}
and hence any such measure $\mu$ defines a frame function $f$ respecting
$\mathbf{PVM}_{d}$, given by $f\left(P_{\mathcal{C}}\right)=\mu\left(\mathcal{C}\right)$.
\end{document}